\documentclass[twocolumn,pra,floatfix,amsmath,superscriptaddress]{revtex4-1}
\usepackage[latin9]{inputenc}
\usepackage{graphicx}
\usepackage[hypertex]{hyperref}
\usepackage{bm,bbm}
\usepackage{amsmath}
\usepackage{amsfonts}
\usepackage{amsthm}
\usepackage{amssymb}
\usepackage{mathrsfs}
\usepackage{subfigure}
\usepackage{array}
\usepackage{color}
\makeatletter

\usepackage{pxfonts}

\newtheorem{lemma}{Lemma}

\newtheorem{observation}{Observation}

\begin{document}

\title{Interpreting quantum coherence through quantum measurement process}

\author{Yao Yao}
\email{yaoyao@mtrc.ac.cn}
\affiliation{Microsystems and Terahertz Research Center, China Academy of Engineering Physics, Chengdu Sichuan 610200, China}
\affiliation{Institute of Electronic Engineering, China Academy of Engineering Physics, Mianyang Sichuan 621999, China}

\author{G. H. Dong}
\affiliation{Beijing Computational Science Research Center, Beijing, 100094, China}

\author{Xing Xiao}
\affiliation{College of Physics and Electronic Information, Gannan Normal University, Ganzhou Jiangxi 341000, China}

\author{Mo Li}
\email{limo@mtrc.ac.cn}
\affiliation{Microsystems and Terahertz Research Center, China Academy of Engineering Physics, Chengdu Sichuan 610200, China}
\affiliation{Institute of Electronic Engineering, China Academy of Engineering Physics, Mianyang Sichuan 621999, China}

\author{C. P. Sun}
\affiliation{Beijing Computational Science Research Center, Beijing, 100094, China}
\affiliation{Synergetic Innovation Center of Quantum Information and Quantum Physics, University of Science and Technology of China,
Hefei, Anhui 230026, China}

\date{\today}

\begin{abstract}
Recently, there has been a renewed interest in the quantification of coherence or other coherence-like concepts
within the framework of quantum resource theory. However, rigorously defined or not,
the notion of \textit{coherence} or \textit{decoherence} has already been used by the community
for decades since the advent of quantum theory. Intuitively, the definitions of coherence and decoherence should be
\textit{the two sides of the same coin}. Therefore, a natural question is raised: how can the conventional decoherence
processes, such as the von Neumann-L\"{u}ders (projective) measurement postulation or partially dephasing channels,
fit into the bigger picture of the recently established theoretic framework? Here we show that the
state collapse rules of the von Neumann or L\"{u}ders-type measurements, as special cases of genuinely incoherent operations (GIO),
are consistent with the resource theories of quantum coherence. New hierarchical measures of coherence are proposed for the L\"{u}ders-type measurement
and their relationship with measurement-dependent discord is addressed.
Moreover, utilizing the fixed point theory for $C^\ast$-algebra, we prove that GIO indeed represent
a particular type of partially dephasing (phase-damping) channels
which have a matrix representation based on the Schur product.
By virtue of the Stinespring's dilation theorem, the physical realizations
of incoherent operations are investigated in detail and we find that GIO in fact constitute the \textit{core} of strictly incoherent
operations (SIO) and generally incoherent operations (IO) and the \textit{unspeakable} notion of coherence induced by GIO can be transferred
to the theories of speakable coherence by the corresponding permutation or relabeling operators.
\end{abstract}


\maketitle
\section{INTRODUCTION}
Quantum coherence, as one of the most fundamental and characteristic concept in quantum theory,
has long been recognized as a valuable resource for modern quantum technologies,
such as quantum computation \cite{Cleve1997,Nielsen}, quantum cryptography \cite{Gisin2002,Scarani2009}
and quantum metrology \cite{Giovannetti2006,Giovannetti2011}. Despite its crucial importance
in the development of quantum information science, only very recently a rigorous theoretical framework
has been established by virtue of quantum resource theory to quantify the usefulness of quantum coherence
contained in quantum states \cite{Baumgratz2014}. In the corresponding resource theory of coherence,
the incoherent (free) states are defined with respect to a \textit{prefixed} orthogonal basis,
which is a convex set containing all diagonal states in this specific basis, while the resource states
are those with nonzero off-diagonal elements. The restricted set of operations (i.e., the
incoherent operations) is constructed with the defining property that every incoherent operation has a Kraus decomposition each branch
of which is coherence-nongenerating \cite{Baumgratz2014}.
Refs. \cite{Streltsov2016,Hu2017} provide detailed reviews of recent advances in the theoretical understanding
and characterization of quantum coherence.

Though such an axiomatic framework is mathematically well-defined, its physical consistency has been further considered \cite{Chitambar2016a,Marvian2016}.
First, the coherence measures proposed in \cite{Baumgratz2014} are apparently basis-dependent and this fact implies that
prior to any usage of these quantifiers, a justification or specification of the choice of basis is needed according to the theoretical model
or experimental setup \cite{Marvian2016,Yao2016}. To be more precise, most recent work based on the resource theory characterizes the \textit{speakable}
notion of coherence \cite{Marvian2016}, that is, relabeling or permutation of basis states is allowed in this occasion, which
is in sharp contrast to the resource theory of unspeakable coherence (i.e., asymmetry) \cite{Marvian2013,Marvian2014}.
Second, several alternative proposals of the resource theory of coherence have also been put forward to impose further constraints on the free operations,
such as the maximal incoherent operations (MIO) \cite{Aberg2006a,Aberg2006b}, dephasing-covariant incoherent operations (DIO) \cite{Marvian2016,Chitambar2016a},
strictly incoherent operations (SIO) \cite{Winter2016,Yadin2016} and genuinely incoherent operations (GIO) \cite{Vicente2017}.
However, the free (i.e., incoherent) operations defined in these scenarios are not \textit{truly} free in the sense of
Stinespring dilation, which means these operations are not strictly free-implementable \cite{Chitambar2016a,Marvian2016}.
Moreover, a physically-consistent resource theory has been introduced in \cite{Chitambar2016a} under the name of
physically incoherent operations (PIO), but the class of PIO is too restrictive and state transformations under this set
are rather limited \cite{Chitambar2016b}.

On the other hand, any realistic quantum system will inevitably interact with its environment and the notion of
decoherence represents the destruction of quantum coherence between a superposition of preferred states \cite{Joos2003,Schlosshauer2007}.
Intuitively, the definitions of coherence and decoherence should be the two sides of the same coin. In comparison to
the resource theory of coherence, the decoherence basis usually emerges associated
with the specific physical process. Two well known examples are the von-Neumann's projective measurement \cite{Neumann} and
the pointer states induced by einselection \cite{Zurek1981,Zurek1982,Zurek2003,Schlosshauer2005}. Moreover, we wonder
whether the resource theory of coherence proposed recently is compatible with previous interpretations of decoherence,
since such a consistency will help us to obtain an in-depth understanding of the paradigmatic models of decoherence processes.
More precisely, the aim of this work is to gain more insight into the characterization of quantum
coherence through the investigations of decohering powers and physical realizations of various types of quantum incoherent operations.

This paper is organized as follows. In Sec. \ref{sec2}, we briefly review two representations of quantum operations
and their relationship. In Sec. \ref{sec3}, we present an interpretation of popular coherence measures through
the von Neumann measurement theory and generalize this line of thought to the L\"{u}ders-type measurement,
where the minimum disturbance principle is highlighted. Moreover, the L\"{u}ders-measurement-dependent discord
is introduced for bipartite system and its relation with L\"{u}ders-type coherences is illustrated.
In Sec. \ref{sec4}, we provide a detailed analysis of the structures and physical realizations of GIO, SIO and IO,
demonstrating that GIO or SIO can be seen as the core of other types of incoherent operations.
Discussions and final remarks are given in Sec. \ref{sec5} and several open questions are raised
for future research.

\section{Stinespring-Kraus representation of quantum channel}\label{sec2}
Let $\mathcal{H}$ be the finite-dimensional Hilbert space and $\mathcal{B}(\mathcal{H})$ ($\mathcal{S}(\mathcal{H})$)
be the set of bounded operators (density operators) on $\mathcal{H}$. A physically valid quantum operation
$\mathcal{E}$: $\mathcal{B}(\mathcal{H})\rightarrow\mathcal{B}(\mathcal{H})$
is defined as a linear trace non-increasing and completely positive map \cite{Kraus1983}.
For simplicity, we assume throughout this paper that $\mathcal{E}$ has equal input and output Hilbert spaces.
In particular, we further identify an operation as a quantum channel if it satisfies the trace-preserving condition.
Mathematically, there exist two explicit and equivalent representations of an arbitrary operation,
which in fact depict the general form of state changes \cite{Stinespring,Kraus1971,Choi1975}:

$\bullet$ The operator-sum representation:
\begin{align}
\mathcal{E}(\rho)=\sum_nK_n\rho K_n^\dagger,
\label{OS}
\end{align}
where $K_n\in \mathcal{B}(\mathcal{H})$, $\sum_nK_n^\dagger K_n\leq\openone_{\mathcal{H}}$
and the equality holds  for quantum channels;

$\bullet$ The Stinespring dilation:
\begin{align}
\mathcal{E}(\rho)=\textrm{Tr}_{\mathcal{A}}(\mathbb{V}\rho \mathbb{V}^\dagger),
\end{align}
where $\mathcal{A}$ is an ancillary system (e.g., an apparatus system) and $\mathbb{V}\in \mathcal{B}(\mathcal{H},\mathcal{H}\otimes \mathcal{A})$
is a contraction (i.e., $\mathbb{V}^\dagger\mathbb{V}\leq\openone_{\mathcal{H}}$). For a trace-preserving map,
$\mathbb{V}$ is actually an isometry.

Intuitively, the Stinespring dilation can be viewed as a \textit{purification} of a quantum operation
on an extended Hilbert space \cite{Buscemi2003}. Furthermore, Kraus and Ozawa proved that a unitary realization can be constructed
for quantum operations, or more generally, quantum instruments \cite{Davies,Kraus1983,Ozawa1984}, which in formula
can be rewritten as
\begin{align}
\mathcal{E}(\rho)=\textrm{Tr}_{\mathcal{A}}\left[(\openone\otimes M_\mathcal{A})\mathbb{U}(\rho\otimes\sigma_{\mathcal{A}})\mathbb{U}^\dagger\right],
\label{IMM1}
\end{align}
where $\mathbb{U}$ is a unitary operation acting on $\mathcal{H}\otimes \mathcal{A}$, $M_\mathcal{A}$
is an effect operator on $\mathcal{A}$ (i.e., $0\leq M_\mathcal{A}\leq\openone_\mathcal{A}$ and $M_\mathcal{A}=\openone_\mathcal{A}$ corresponds to quantum channels)
and $\sigma_{\mathcal{A}}$ is the initial state of the apparatus system.
Eq. (\ref{IMM1}) shows that for a particular quantum operation $\mathcal{E}$, the four-tuple $\{\mathcal{A},\sigma_{\mathcal{A}},\mathbb{U},M_\mathcal{A}\}$
uniquely determines the state change caused by $\mathcal{E}$. In other words, the four-tuple
provides a physical realization of the operation. Under different names, such a realization is also known as the system-apparatus interaction \cite{Neumann},
premeasurement \cite{Beltrametti1990}, or indirect measurement model \cite{Ozawa2000}.

In addition, without loss of generality, one may require that $\sigma_{\mathcal{A}}$ is a pure state and
$M_\mathcal{A}$ is an orthogonal projection operator \cite{Kraus1983}.
Therefore, by denoting $\sigma_{\mathcal{A}}=|a_0\rangle\langle a_0|$ and $M_\mathcal{A}=P_\mathcal{A}$, Eq. (\ref{IMM1}) can be reexpressed as
\begin{align}
\mathcal{E}(\rho)=\textrm{Tr}_{\mathcal{A}}\left[(\openone\otimes P_\mathcal{A})\mathbb{U}(\rho\otimes|a_0\rangle\langle a_0|)\mathbb{U}^\dagger\right].
\label{IMM2}
\end{align}
To see the direct correspondence between two representations of Eqs. (\ref{OS}) and (\ref{IMM2}), it is convenient to
specify an orthogonal decomposition of $P_\mathcal{A}=\sum_n|a_n\rangle\langle a_n|$ and hence
the Kraus operators can be expressed as \cite{Nielsen}
\begin{align}
K_n=\langle a_n|\mathbb{U}|a_0\rangle.
\end{align}
Moreover, the non-uniqueness of Kraus decomposition can be regarded as stemming from the freedom in choosing the basis $\{|a_n\rangle\}$.
Hence different sets of Kraus operators are related to each other by isometric matrices.

\section{L\"{u}ders-type quantum coherence}\label{sec3}
In his seminal work \cite{Neumann}, von Neumann pointed that, in contrast to the unitary transformations described by the Schr\"{o}dinger equation,
there exists another type of intervention for quantum systems. In fact, he formulated a measurement and state-reduction process with respect to
purely discrete and nondegenerate observables, which is better known as the \textit{state collapse postulate}. Later,
L\"{u}ders generalized von Neumann's postulate to degenerate observables \cite{Luders}. In this section, we connect the
von Neumann-L\"{u}ders measurement theory to the interpretation and characterization of the coherence contained
in quantum states, especially relative to the observables under consideration.

\subsection{von Neumann-L\"{u}ders measurement postulation}
Let $\rho\in\mathcal{S}(\mathcal{H})$ be a density matrix of a quantum system in Hilbert space $\mathcal{H}$ and $R$ be a discrete, non-degenerate observable
with the eigen-decomposition $R=\sum_nr_n|\phi_n\rangle\langle\phi_n|$. Based on the Compton-Simons experiment,
von Neumann derived the the well-known state collapse postulate by virtue of the following statistical rule and hypothesis \cite{Neumann}:

$\bullet$ \textit{Born's statistical rule}, which demands that the probability for obtaining the measurement
result $r_n$ is given by
\begin{align}
P(r_n)=\textrm{Tr}(\rho|\phi_n\rangle\langle\phi_n|)=\langle\phi_n|\rho|\phi_n\rangle.
\end{align}
Note that this formula can be generalized to more general measurements described by positive operator-valued measures (POVM)
$\mathcal{M}=\{M_n\}$ with $M_n\geq0$ and $\sum_nM_n=\openone$. Namely, the probability of obtaining the outcome $n$ is $P(M_n)=\textrm{Tr}(\rho M_n)$ \cite{Nielsen}.

$\bullet$ \textit{Repeatability hypothesis}, which states that if a physical quantity is measured twice in succession in a system, then we
get the same value each time. This hypothesis is equivalent to a requirement on the conditional probability:
\begin{align}
P(r_m|r_n)=\textrm{Tr}(\rho_n|\phi_m\rangle\langle\phi_m|)=\delta_{mn},
\end{align}
where $\rho_n$ is the (normalized) resulting state of the system after obtaining the measurement outcome $r_n$.

In particular, according to the repeatability hypothesis, it is easy to prove that the eigenstate $|\phi_n\rangle\langle\phi_n|$ of
the observable $R$ is the only possible post-measurement state for the outcome $r_n$ (see Appendix \ref{A1}). Therefore,
the density matrix $\rho$ is transformed to the following statistical mixture
\begin{align}
\sigma=\sum_nP(r_n)\rho_n=\sum_n\langle\phi_n|\rho|\phi_n\rangle|\phi_n\rangle\langle\phi_n|.
\end{align}
In the language of quantum operation, The corresponding change of the state can be represented by
\begin{align}
\mathcal{D}(\bullet)=\sum_n|\phi_n\rangle\langle\phi_n|\bullet|\phi_n\rangle\langle\phi_n|,
\end{align}
where this superoperator is also known as (completely) pure-dephasing channel or pinching operator \cite{Bhatia}.
Note that $\mathcal{D}$ is idempotent (i.e., $\mathcal{D}^2=\mathcal{D}$) and retains only the diagonal elements of the density matrix.
Apart from the elimination of the off-diagonal elements, it is noteworthy that the decoherence effect of $\mathcal{D}$
is also manifested in the increase of von Neumann entropy \cite{Neumann} (see also Appendix \ref{A2}).

Since the initial state is \textit{completely} decohered by a von Neumann measurement, the above two signatures of decoherence
can be employed to quantifying quantum coherence contained in states. In fact, prior to the rigorous definitions of
quantum coherence in Ref. \cite{Baumgratz2014}, the magnitude of off-diagonal elements in certain basis has long been recognized
as a convenient and useful quantifier of coherence, for instance, in the discussion of quantum interferometric complementary \cite{Jaeger1995,Durr2001,Englert2008}.
Moreover, the von Neumann entropy produced by the projective measurement, dubbed as \textit{the entropy of coherence},
has also been proposed in an attempt to quantify the incompatibility
between a given (nondegenerate) observable and a given quantum state \cite{Herbut2002,Herbut2005},
which is exactly the entropic measure of coherence defined in \cite{Baumgratz2014}. Mathematically,
if we define the set of incoherent states with respect to the nondegenerate observable $R$ as
\begin{align}
\mathcal{I}(\mathcal{H})=\{\rho: \mathcal{D}(\rho)=\rho, \rho\in\mathcal{S}(\mathcal{H})\},
\end{align}
then the corresponding measures of coherence can be formulated as
\begin{align}
C_{l_1}(\rho)&=\sum_{m\neq n}|\langle\phi_m|\rho|\phi_n\rangle|=\min_{\sigma\in\mathcal{I}(\mathcal{H})}\|\rho-\sigma\|_{l_1},\\
C_{\textrm{re}}(\rho)&=S(\mathcal{D}(\rho))-S(\rho)=\min_{\sigma\in\mathcal{I}(\mathcal{H})}S(\rho\|\sigma).
\end{align}

On the other hand, if the eigen-decomposition of the observable $R=\sum_nr_nP_n$ is degenerate (i.e., $d_n=\textrm{Tr}P_n\geq 1$ denote degeneracies),
von Neumann's theory still follows the same routine by alternatively measuring
a commuting \textit{fine-grained} observable $\mathfrak{R}=\sum_{ni}\mu_{ni}|\phi_{ni}\rangle\langle\phi_{ni}|$,
where $\sum_{i=1}^{d_n}|\phi_{ni}\rangle\langle\phi_{ni}|=P_n$ and $\langle\phi_{mi}|\phi_{nj}\rangle=\delta_{mn}\delta_{ij}$.
By defining a function $f$ with $f(\mu_{ni})=r_n$ for all $i=1,\ldots,d_n$, the above fine-graining process can be encapsulated in the following
\begin{align}
R=f(\mathfrak{R}).
\label{fine-graining}
\end{align}

However, since there exists infinite number of ways to decompose the degenerate eigenspaces, this apparent arbitrariness
would lead to the non-uniqueness of state transformation, which means that the formula of state change
will depend on the specific choice of $\widetilde{R}$. To avoid the ambiguousness, L\"{u}ders
generalized von Neumann's postulate to degenerate observables by introducing an extended ansatz for state reduction,
that is \cite{Luders}
\begin{align}
\mathcal{L}(\rho)=\sum_{n}P_n\rho P_n,
\label{Luders}
\end{align}
where $\mathcal{L}(\bullet)$ is also known as the L\"{u}ders state transformer or L\"{u}ders instrument \cite{Busch1991}.
Remarkably, except for the hypothesis of discreteness of spectrum and repeatability, it is demonstrated that
the L\"{u}ders-type state transformation can be derived by introducing an additional requirement of \textit{least interference}
or \textit{minimal disturbance} \cite{Goldberger1964,Herbut1969}. Indeed, by defining a generalized set of incoherent states
\begin{align}
\mathfrak{I}(\mathcal{H})=\{\rho: \mathcal{L}(\rho)=\rho, \rho\in\mathcal{S}(\mathcal{H})\},
\end{align}
it can be shown that the repeatability hypothesis \textit{alone} would render the (possible) reduced state $\sigma$
belonging to $\mathfrak{I}(\mathcal{H})$ (see Appendix \ref{A1}). Form the geometric point of view, the principle of
minimal disturbance amounts to the requirement $\sigma$ is \textit{closest} to the initial state $\rho$ and hence
\textit{uniquely} determines the change-of-state formula.
Thus, the distance metrics, such as matrix norms or entropy quantities, can be unitized to measure the degree of closeness.

In particular, the Hilbert-Schmidt norm $\|\bullet\|_2$ turns out to be a potential choice for demonstrating the closeness
due to its explicit physical meaning and convenience (e.g., basis-independent) \cite{Herbut1969}.
By using the properties of the Hilbert-Schmidt norm, we have
\begin{align}
\|\rho-\sigma\|_2^2&=\left\|\sum_{m\neq n}P_m\rho P_n+\sum_n(P_n\rho P_n-P_n\sigma P_n)\right\|_2^2   \nonumber\\
&=\sum_{m\neq n}\left\|P_m\rho P_n\right\|_2^2+\sum_n\left\|P_n\rho P_n-P_n\sigma P_n\right\|_2^2.
\end{align}
To obtain the minimum value of $\|\rho-\sigma\|_2$, every term in the second summation should be equal to zero, which
is equivalent to the condition $P_n\sigma P_n=P_n\rho P_n$ for all $n$. Therefore, the formula of state change
(i.e., the L\"{u}ders state transformer) can be uniquely determined as
\begin{align}
\sigma=\sum_{n}P_n\sigma P_n=\sum_{n}P_n\rho P_n,
\end{align}
which is exactly the Eq. (\ref{Luders}).

In fact, the above argument can also be extended to the quantum relative entropy, another important quantity
in quantum information theory. Using the idempotent property of projectors, the cyclic property of trace
and the commutation relation $[\sigma,R]=0$, one can obtain the following inequality
\begin{align}
S(\rho\|\sigma)&=S(\rho\|\mathcal{L}(\rho))+S(\mathcal{L}(\rho)\|\sigma)   \nonumber\\
&\geq S(\rho\|\mathcal{L}(\rho)),
\end{align}
where the equality holds for $\sigma=\mathcal{L}(\rho)$. The $l_1$ norm may also participate but
it is a little bit cumbersome since the $l_1$ norm is basis-dependent. Here we can borrow the same
idea from von Neumann that one can decompose the set of orthogonal projectors $\{P_n\}$ into
an \textit{bi-orthogonal} basis $\{\phi_{ni}\}$ for $n=1,\ldots,N$ and $i=1,\ldots,d_n$, where
the dimension of Hilbert space $d=\sum_nd_n\geq N$. For a particular choice of basis $\{\phi_{ni}\}$,
the argument is similar to that of Hilbert-Schmidt norm
\begin{align}
\|\rho-\sigma\|_{l_1}&=\left\|\sum_{m\neq n}P_m\rho P_n+\sum_n(P_n\rho P_n-P_n\sigma P_n)\right\|_{l_1}  \nonumber\\
&=\sum_{m\neq n}\left\|P_m\rho P_n\right\|_{l_1}+\sum_n\left\|P_n\rho P_n-P_n\sigma P_n\right\|_{l_1},
\end{align}
where in such a decomposition of eigenspaces the $l_1$ norm is calculated independently.
Therefore, the above derivations present an alternative and straightforward interpretation of the framework of the L\"{u}ders measurement,
\textit{from the perspective of coherence theory}: while the repeatability hypothesis induces a block-diagonal structure of the state reduction,
the principle of least interference or minimal disturbance is equivalent to the requirement
that the von Neumann-L\"{u}ders measurement will always lead to a final state
which is \textit{closest} to the initial state, comparing to all the other states with
no (generalized) coherence in corresponding decomposition of Hilbert space.
In this sense, the von Neumann-L\"{u}ders measurement is usually deemed as
a completely decohering (or dephasing) channel in the framework of quantum coherence.
Hence, in the resource theory of coherence, the completely decohering (or dephasing) channel
serves as a basic reference for other types of incoherent operations \cite{Chitambar2016b}.

\subsection{Coarse-graining of quantum coherence}
Based on the above geometric considerations, we can generalize the measures of coherence
for the non-degenerate observable to the L\"{u}ders-type measurement. With respect to
the spectral decomposition of a degenerate observable $R=\sum_nr_nP_n$, we define
\begin{align}
C_{l_1}(R,\rho)&=\min_{\sigma\in\mathfrak{I}(\mathcal{H})}\|\rho-\sigma\|_{l_1}=\sum_{m\neq n}\|P_m\rho P_n\|_{l_1},\\
C_{\textrm{re}}(R,\rho)&=\min_{\sigma\in\mathfrak{I}(\mathcal{H})}S(\rho\|\sigma)=S(\mathcal{L}(\rho))-S(\rho).
\end{align}
Note that when $R$ is nondegenerate the generalized set of incoherent states $\mathfrak{I}(\mathcal{H})$
reduces to the ordinary set $\mathcal{I}(\mathcal{H})$.
It is worth to emphasize again that $C_{l_1}(R,\rho)$ is a basis-dependent quantity, where a particular
orthogonal decomposition of eigen-projectors $\{P_n\}$ should be specified, for example, a fine-graining observable
$\mathfrak{R}$ in Eq. (\ref{fine-graining}). On the contrary, $C_{\textrm{re}}(R,\rho)$ is irrespective of such a
fine-graining and hence more feasible and convenient. Thus, $C_{l_1}(R,\rho)$ and $C_{\textrm{re}}(R,\rho)$
can be viewed as a \textit{coarse-graining} version of the corresponding measures proposed in \cite{Baumgratz2014},
and the coarse-graining process is also manifested by the hierarchy relation
\begin{align}
C_{l_1}(\mathfrak{R},\rho)\geq C_{l_1}(R,\rho),\quad
C_{\textrm{re}}(\mathfrak{R},\rho)\geq C_{\textrm{re}}(R,\rho).
\label{Coh-inequality}
\end{align}

Since the first inequality is easily proved by using the relation $\sum_{n\neq m}\sum_{i,j}\leq\sum_{ni\neq mj}$,
the second inequality can be verified by the identity
\begin{align}
C_{\textrm{re}}(\mathfrak{R},\rho)-C_{\textrm{re}}(R,\rho)=S\left[\mathcal{L}_{R}(\rho)\|\mathcal{D}_{\mathfrak{R}}(\rho)\right]\geq0,
\end{align}
where we attach suffixes $R$ and $\mathfrak{R}$ to super-operators $\mathcal{L}$ and $\mathcal{D}$ respectively to indicate with respect to which observable
the corresponding measurement is preformed and note that $\mathcal{L}_{\mathfrak{R}}=\mathcal{D}_{\mathfrak{R}}$ since $\mathfrak{R}$ is nondegenerate.
The differences in Eq. (\ref{Coh-inequality}) indicate that the L\"{u}ders measurement retains some \textit{residual coherence} which resides in
every block of $\mathcal{L}_{R}(\rho)$.
Intriguingly, it was proved that for any state $\rho$ and any degenerate observable $R$ there exists (at least) one fine-grained nondegenerate observable $\mathfrak{R_\star}$
(i.e., $R=f(\mathfrak{R_\star})$) satisfying $\mathcal{L}_{R}=\mathcal{D}_{\mathfrak{R_\star}}$ \cite{Herbut1974}. In this case,
we have
\begin{align}
C_{l_1}(\mathfrak{R}_\star,\rho)=C_{l_1}(R,\rho),\quad
C_{\textrm{re}}(\mathfrak{R}_\star,\rho)=C_{\textrm{re}}(R,\rho),
\end{align}
where the $l_1$ norm of coherence is defined with respect to the common eigenvectors of $R$ and $\mathfrak{R_\star}$ and note that
\begin{align}
C_{l_1}(R,\rho)&=\|\rho-\mathcal{L}_R(\rho)\|_{l_1},\\
C_{l_1}(\mathfrak{R_\star},\rho)&=\|\rho-\mathcal{D}_\mathfrak{R_\star}(\rho)\|_{l_1}.
\end{align}

Moreover, it is natural to extend our consideration to the multipartite system. Consider a bipartite state
$\rho^{AB}$ with reduced states $\rho^A$, $\rho^B$ and a L\"{u}ders measurement of observable $R=\sum_nr_nP_n$
on subsystem B. One can define an \textit{observable-dependent} version of quantum-incoherent (QI) states of the form
\begin{align}
\chi^{AB}=\mathcal{L}^B(\rho^{AB})=\sum_n(\openone\otimes P_n)\rho^{AB}(\openone\otimes P_n),
\end{align}
which would reduce to the normal QI-states introduced in Ref. \cite{Chitambar2016c} if $R$ is nondegenerate.
Note that for degenerate observables (i.e., $\textrm{Tr}P_n>1$ for some index $n$), $\chi^{AB}$ may be entangled,
which is in sharp contrast to the case of von Neumann measurement \cite{Luo2013}.
The generalized QI relative entropy of coherence can be defined as
\begin{align}
C_\textrm{re}^{A|B}(R,\rho^{AB})&=\min_{\chi^{AB}\in\mathcal{QI}}S(\rho^{AB}\|\chi^{AB})  \nonumber\\
&=S(\mathcal{L}^B(\rho^{AB}))-S(\rho^{AB}),
\end{align}
where $\mathcal{QI}$ denotes the set of observable-dependent QI-states.

Inspired by the concept of the basis-dependent quantum discord (i.e., discord dependent on a particular von Neumann measurement)
\cite{Ollivier2001,Henderson2001},
one can define a similar observable-dependent measure of quantum discord
\begin{align}
\delta^{A|B}(R,\rho^{AB})=I(\rho^{AB})-I(\mathcal{L}^B(\rho^{AB})),
\end{align}
where $I(\rho^{AB})=S(\rho^{AB}\|\rho^{A}\otimes\rho^{B})$ is the quantum mutual information of $\rho^{AB}$.
Remarkably, the L\"{u}ders-type quantum discord $\delta^{A|B}(R,\rho^{AB})$ is closely related to the L\"{u}ders-type coherences.
Indeed, a simple algebra shows that
\begin{align}
\delta^{A|B}(R,\rho^{AB})&=C_\textrm{re}^{A|B}(R,\rho^{AB})-C_\textrm{re}(R,\rho^{B}).
\label{Coh-Discord}
\end{align}
When the observable is nondegenerate, then $R$ specifies a orthogonal basis and Eq. (\ref{Coh-Discord})
recovers the same relation for von Neumann measurement \cite{Yadin2016}.
Notably when $R$ is degenerate, the L\"{u}ders-type quantum discord is highly nontrivial \cite{Luo2013}.
In fact, the observable-dependent classical correlation can be defined as
\begin{align}
\mathcal{J}^{A|B}(R,\rho^{AB})&=I(\rho^{AB})-\delta^{A|B}(R,\rho^{AB}) \nonumber\\
&=\sum_np_nS(\rho^A_n\|\rho^A)+\sum_np_nI(\rho^{AB}_n),
\label{OD-CC}
\end{align}
with the post-measurement state $\rho^{AB}_n=(\openone\otimes P_n)\rho^{AB}(\openone\otimes P_n)$
and $\rho^{A}_n=\textrm{Tr}\rho^{AB}_n$. It is worth noting that the second term in Eq. (\ref{OD-CC})
is missing in the original definition of classical correlation \cite{Ollivier2001,Henderson2001}
since for the von Neumann measurement $\rho^{AB}_n$ is a product state. However, for the L\"{u}ders measurement,
we may have $I(\rho^{AB}_n)>0$ implying $\rho^{AB}_n$ is not factorable. This residual part also reflects the fact that
the L\"{u}ders measurement is more gentle than von Neumann measurement and maintains partial coherence
in the measurement process. Interesting, very recently, the author of Ref. \cite{Luo2013},
presented two related papers \cite{Luo2017a,Luo2017b},
where the significance of the L\"{u}ders measurement has also been highlighted in the characterization of quantum coherence using
the skew information.

\section{GIO as Partially dephasing channels}\label{sec4}
For a proper choice of the orthogonal basis, the von Neumann or L\"{u}ders measurement can also be viewed as
special cases of GIO, which are an essential subset of quantum channels preserving all incoherent basis states  \cite{Vicente2017}. By definition,
a crucial fact is that GIO lead to an \textit{unspeakable} notion of quantum coherence within the framework of resource theory \cite{Marvian2016},
which means permutation or relabeling is not allowed regarding the state transformations induced by GIO.
To gain a deeper insight into the nature of GIO, we initiate a further analysis of GIO from two different perspectives:
one from the fixed-point theory of quantum maps and the other from the physical realization of GIO, both highlighting
that GIO are at the core of the resource theory of quantum coherence.

\subsection{Fixed points of unital quantum channels}
Here we consider a finite $d$-dimensional Hilbert space and a CPTP map (i.e., quantum channel) $\Phi:\mathcal{B}(\mathcal{H})\rightarrow\mathcal{B}(\mathcal{H})$.
The property of complete positivity guarantees that $\Phi(\bullet)$  has an operator-sum representation of the form $\Phi(\bullet)=\sum_iK_i\bullet K_i^\dagger$,
and trace preservation of $\Phi$ is equivalent to $\sum_iK_i^\dagger K_i=\openone$.
For a prefixed orthogonal basis $\{|\phi_n\rangle\}$ (or with respect to a nondegenerate observable $R=\sum_nr_n|\phi_n\rangle\langle\phi_n|$),
the set of GIO can be proposed with the defining property
\begin{align}
\textrm{GIO}=\{\Phi:\Phi(\rho)=\rho,\rho\in\mathcal{I}(\rho)\}.
\end{align}
By the linearity of $\Phi$, the above definition is tantamount to
\begin{align}
\textrm{GIO}=\{\Phi:\Phi(|\phi_n\rangle\langle\phi_n|)=|\phi_n\rangle\langle\phi_n|,\forall n\},
\end{align}
which implies that pure incoherent basis states are \textit{fixed points} for GIO. Obviously,
the identity matrix $\openone$ is also preserved by GIO and hence GIO are unital quantum channels (i.e., $\sum_iK_i K_i^\dagger=\openone$).

According to Schauder's fixed-point theorem, there exists at least one density matrix $\rho$ for a CPTP map such that
$\Phi(\rho)=\rho$ \cite{Granas}. Indeed, fixed point theory has already been employed in the investigations of
quantum error correction \cite{Raginsky2002,Kribs2005,Blume-Kohout2008} and quantum reference frame \cite{Gour2009}.
To proceed, we need to introduce the notion of the (noise) commutant of the matrix algebra generated by the set of Kraus operators
$\{K_i,K_i^\dagger\}$, that is
\begin{align}
\mathcal{A}'=\{X\in\mathcal{B}(\mathcal{H}): [X,A]=0, A\in\{K_i,K_i^\dagger\}, \forall i\}.
\end{align}
It is easy to see that $\mathcal{A}'\subseteq \mathcal{F}(\Phi)$,
where $\mathcal{F}(\Phi)=\{X\in\mathcal{B}(\mathcal{H}): \Phi(X)=X\}$ denotes the set of fixed points of
\textit{unital} channel $\Phi$. Notably, the converse inclusion relation is also true, in other words,
 $\Phi$ we have the following lemma \cite{Arias2002,Kribs2003,Holbrook2003}
\begin{lemma}
For a (finite-dimensional) unital quantum channel $\Phi$, we have $\mathcal{A}'=\mathcal{F}(\Phi)$.
\label{C=F}
\end{lemma}
\begin{proof}
Note that in our case the converse inclusion relation can be elegantly proved by the identity \cite{Lindblad1999,Watrous}
\begin{align}
\sum_i[X,K_i][X,K_i]^\dagger=\Phi(XX^\dagger)-XX^\dagger,
\end{align}
with the trace-preserving property of $\Phi$.
\end{proof}

Now we present our first key observation:

\begin{observation}
The function of GIO is fully characterized by a correlation matrix $\mathcal{C}$, which can be represented
as a Gram matrix of a set of dynamical vectors $\{|c_i\rangle\}_{i=1}^d$.
\end{observation}

\begin{proof} Applying Lemma \ref{C=F} to GIO, we now know that every Kraus operator of GIO commutes with all incoherent basis states,
indicating that all the Kraus operators must be of diagonal form with respect to the incoherent basis
\begin{align}
K_i=\sum_{j=1}^dc_{j}^{(i)}|\phi_j\rangle\langle\phi_j|,\, \forall i=1,\ldots,r,
\end{align}
with $r$ being the Choi rank of $\Phi$. From $\sum_iK_i^\dagger K_i=\openone$, we note that
\begin{align}
\sum_iK_i^\dagger K_i=\sum_j\sum_i|c_{j}^{(i)}|^2|\phi_j\rangle\langle\phi_j|=\openone,
\end{align}
which implies that the vectors $|c_i\rangle=(c_i^{(1)},c_i^{(2)},\ldots,c_i^{(r)})$ are automatically normalized.
Furthermore, the function of $\Phi$ can be represented as a \textit{Schur product} (i.e., entry-wise product) of the form
\begin{align}
\Phi(\rho)=\sum_iK_i\rho K_i^\dagger=\mathcal{C}^{T}\circ\rho,
\end{align}
where we define the correlation matrix as
\begin{align}
\mathcal{C}=
\left(\begin{matrix}
1 & \langle c_1|c_2\rangle & \ldots & \langle c_1|c_r\rangle  \\
\langle c_2|c_1\rangle & 1 & \ldots & \langle c_2|c_r\rangle  \\
\vdots & \vdots & \ddots & \vdots  \\
\langle c_r|c_1\rangle & \langle c_r|c_2\rangle & \ldots & 1  \\
\end{matrix}\right),
\end{align}
and the Schur (Hadamard) product of $A=[a_{ij}]$ and $B=[b_{ij}]$ is denoted by $A\circ B=[a_{ij}b_{ij}]$.
\end{proof}

Intriguingly, since the correlation matrix $\mathcal{C}$ is a Gram matrix of a set of vectors $\{|c_i\rangle\}_{i=1}^d$,
thus $\mathcal{C}$ is a positive semi-definite matrix, which confirms the positivity of $\Phi$ by Schur product theorem (Theorem 5.2.1 in \cite{Horn1991}).
Note that $\mathcal{C}$ is uniquely determined by $\Phi$ and the entries on the main diagonal always equal to 1.
As special cases of GIO, the von Neumann and L\"{u}ders measurement can be recast as
\begin{align}
\mathcal{D}(\rho)=\openone\circ\rho,\,
\mathcal{L}(\rho)=E\circ\rho,
\end{align}
where $E=E_{d_1}\oplus\cdots\oplus E_{d_N}$ with $d=\sum_{n=1}^Nd_n$ and $E_{d_n}$ denotes the $d_n$-dimensional square matrix with all
entries equal to 1. Another important example is the phase-damping channel
$\mathcal{E}(\rho)=p\rho+(1-p)\sigma_z\rho\sigma_z$ for a qubit system, which can also be written as a Schur product
\begin{align}
\mathcal{E}(\rho)
=\left(\begin{array}{cc}
1 & 2p-1 \\
2p-1 & 1
\end{array}\right)\circ\rho,
\end{align}
where $p\in[0,1]$ is the noise parameter.

If the decoherence basis $\{|\phi_i\rangle\}$ is fixed,
the decoherence effect is explicitly exhibited by the decay of the absolute value of matrix elements since
$|C_{ij}|=|\langle c_i|c_j\rangle|\leq1$. Moreover, this decoherence effect can be clearly seen
through successive uses of the channel
\begin{align}
\lim_{n\rightarrow\infty}\Phi^n(\rho)=\mathcal{D}(\rho).
\end{align}
Moreover, the entropy increase of GIO can be verified by the following majorization relation \cite{Bapat1985}
\begin{align}
\bm{\lambda}(A\circ B)\prec\bm{\lambda}(A)\circ\bm{\lambda}(\mathcal{D}(B))\prec\bm{\lambda}(A)\circ\bm{\lambda}(B),
\end{align}
where $A\geq0$, $B\geq0$ and $\bm{\lambda}(X)$ denotes the vector of eigenvalues of matrix $X$ in decreasing order.
If we choose $A=\rho$, $B=\mathcal{C}^{T}$ and note that $\mathcal{D}(\mathcal{C})=\openone$, we obtain
\begin{align}
\bm{\lambda}(\Phi(\rho))\prec\bm{\lambda}(\rho),
\end{align}
which leads to the inequality $S(\Phi(\rho))\geq S(\rho)$ for GIO \cite{Marshall} (see Appendix \ref{A2} for more discussion).

\subsection{Physical realization of GIO}
In his seminal work, von Neumann introduced a description of quantum measurement process for discrete observables in terms of
the interaction between system and apparatus \cite{Neumann}. Later, Ozawa generalized this description to continuous observables
in the framework of quantum instrument \cite{Davies}, where a four-tuple $\{\mathcal{A},|a_0\rangle,\mathbb{U},P_\mathcal{A}\}$
is proposed to fully characterize a measuring process \cite{Ozawa1984,Ozawa2000}. In such an indirect-measurement model,
the interaction unitary operator $\mathbb{U}$ plays a central role in establishing the correlation between the observed system and
the measuring apparatus. For instance, in von Neumann's premeasurement of a observable $R=\sum_nr_n|\phi_n\rangle\langle\phi_n|$ with
nondegenerate eigenvalues $r_n$, the structure of $\mathbb{U}$ is determined by
\begin{align}
\mathbb{U}_{\textrm{N}}\left(|\phi_n\rangle\otimes|a_0\rangle\right)=|\phi_n\rangle\otimes|a_n\rangle,
\end{align}
where $|a_0\rangle$ is a fixed pure state in the Hilbert space $\mathcal{A}$ of the apparatus system
and $\{|a_n\rangle\}$ is an orthogonal basis in $\mathcal{A}$. Hence if we measure an
observable $M_\mathcal{A}=\sum_nr_n|a_n\rangle\langle a_n|$ on the apparatus system,
a perfect correlation of measurement outcomes between $R$ and $M_\mathcal{A}$ will be established by $\mathbb{U}_{\textrm{N}}$
and the repeatability of von Neumann measurement is guaranteed \cite{Beltrametti1990}.
For the L\"{u}ders measurement, $\mathbb{U}$ admits a similar structure and the degeneracy of $R=\sum_nr_nP_n$
is taken into account
\begin{align}
\mathbb{U}_{\textrm{L}}\left(|\phi_{ni}\rangle\otimes|a_0\rangle\right)=|\phi_{ni}\rangle\otimes|a_n\rangle,
\end{align}
where $|\phi_{ni}\rangle$ constitute an orthonormal basis of $\mathcal{H}$ such that $P_n=\sum_i|\phi_{ni}\rangle\langle \phi_{ni}|$.
Obviously, the Kraus operator is exactly the orthogonal projector, i.e., $K_n=\langle a_n|\mathbb{U}_{\textrm{L}}|a_0\rangle=P_n$
and correpondingly the formula of state change is $\mathcal{L}(\rho)=\sum_nP_n\rho P_n$.

Since the von Neumann and L\"{u}ders measurement are special cases of GIO, intuitively $\mathbb{U}$ for GIO should have some extra degree of freedom
in its construction. Indeed, with respect to a complete orthogonal basis $\{\phi_{n}\}$
the interaction unitary operator $\mathbb{U}$ for GIO would be of the form
\begin{align}
\mathbb{U}_{\textrm{GIO}}\left(|\phi_{n}\rangle\otimes|a_0\rangle\right)=|\phi_{n}\rangle\otimes|c_n\rangle,
\end{align}
where $|c_n\rangle$ is exactly the one defined in the above subsection, that is, $c_n=\sum_ic^{(i)}_n|a_i\rangle$.
To gain a deeper insight, we have the following remarkable observation:
\begin{observation}
 $\mathbb{U}_{\textrm{GIO}}$ can be represented as a controlled-unitary operation, namely,
 \begin{align}
\mathbb{U}_{\textrm{GIO}}=\sum_n|\phi_{n}\rangle\langle\phi_{n}|\otimes U_n.
\end{align}
\end{observation}
The effect of $U_n$ is to transform the fixed pure state $|a_0\rangle$ to a normalized vector $|c_n\rangle$, but \textit{not
necessary} orthogonal for distinct $n$. For comparison, when $\mathbb{U}_{\textrm{GIO}}$ reduces to $\mathbb{U}_{\textrm{N}}$
the set of $\{U_n\}$ transforms $|a_0\rangle$ to a complete set of orthogonal basis $\{|a_n\rangle\}$.
The corresponding Kraus operators are consistent with the previous discussion since
\begin{align}
K_n=\langle a_n|\mathbb{U}_{\textrm{GIO}}|a_0\rangle=\sum_ic^{(n)}_i|\phi_{i}\rangle\langle\phi_{i}|.
\end{align}

Especially, another significant example of controlled-unitary operations is the generalized controlled-NOT (CNOT) gate,
which can be defined by \cite{Daboul2003}
\begin{align}
\mathbb{U}_{\textrm{CNOT}}=\sum_{n=1}^{d}|n\rangle\langle n|\otimes \mathbb{X}^n,
\end{align}
where $\mathbb{X}$ is the generalized Pauli operator with $\mathbb{X}|i\rangle=|i+1\pmod d\rangle$ and
$d=\min(d_S,d_A)$ with $d_S$ ($d_A$) being the dimension of Hilbert space of the system (apparatus).
Note that the CNOT gate is a key ingredient for connecting resource theories of entanglement
to that of quantum coherence \cite{Streltsov2016} and is itself a bipartite SIO \cite{Winter2016}.
In contrast, $\mathbb{U}_{\textrm{GI}}$ is only incoherent with respect to the observed system
but overall coherence-generating for the system-apparatus interaction (e.g., $\{|a_n\rangle\}$
is chosen to the incoherent basis for the apparatus).

\subsection{Dissecting the structure of SIO and IO}
To illustrate the GIO as the core of SIO and IO, we first recall the relevant definitions and properties
of various types of incoherent operations. In the context of IO, the constraint of \textit{coherence-non-generating}
is put on the set of Kraus operators, which corresponds to a specific physical realization of IO \cite{Baumgratz2014}.
Accordingly, the notion of MIO is defined by putting the same constraint on its overall operation, irrespective of
the specific Kraus decomposition \cite{Aberg2006a,Chitambar2016b}. Along the same line, the relationship between
SIO and DIO is similar to that of IO and MIO, but the constraint is substituted by \textit{coherence-non-exploiting}
for a classical observer \cite{Winter2016,Yadin2016}. Interestingly, though the constraint of \textit{incoherent-state-preserving}
is the defining property of GIO, it has been shown that this constraint is automatically satisfied by
every Kraus decomposition of GIO. Indeed, this phenomenon has its root in the fact that GIO introduce a notion of
unspeakable coherence while SIO and IO are resource theories of speakable coherence \cite{Marvian2016,Vicente2017}.

Moreover, it has been rigorously proved in our previous work that the constraint of coherence-non-generating
(i.e., mapping every incoherent state to an incoherent state) would render every Kraus operator of IO to admit
the following representation \cite{Yao2015}
\begin{align}
K_n^{\textrm{IO}}=\sum_ic^{(n)}_i|f_n(i)\rangle\langle i|,
\end{align}
with $f_n(i)$ being a relabeling function specified by index $n$.
This structure guarantees that there exists at most one nonzero entry in every
column of $K_n^{\textrm{IO}}$.
Furthermore, SIO require that its dual operation would also satisfy this constraint, that is,
$K_n^\dagger\mathcal{I}(\mathcal{H})K_n\subseteq\mathcal{I}(\mathcal{H})$, which implies
\begin{align}
K_n^{\textrm{SIO}}=\sum_ic^{(n)}_i|\pi_n(i)\rangle\langle i|,
\end{align}
with $\pi_n(i)$ being a permutation function specified by index $n$.
Note that there is a crucial difference between $f_n(i)$ and $\pi_n(i)$: $\pi_n(i)$ is
bijective and invertible but in general $f_n(i)$ may \textit{not} be injective.
Therefore, the following observation is straightforward concerning this distinction:

\begin{observation}
The Kraus operators of SIO and IO can be obtained by combining Kraus operators of GIO
with the permutation operator and relabeling operator respectively. Mathematically,
we have
\begin{align}
K_n^{\textrm{SIO}}=\mathcal{P}_nK^{\textrm{GIO}}_n, \,
K_n^{\textrm{IO}}=\mathcal{R}_nK^{\textrm{GIO}}_n,
\end{align}
where we define
\begin{align}
\mathcal{P}_n=\sum_i|\pi_n(i)\rangle\langle i|, \,
\mathcal{R}_n=\sum_i|f_n(i)\rangle\langle i|.
\end{align}
\end{observation}

Note that the \textit{permutation operator} $\mathcal{P}_n$ is in fact a unitary incoherent operator. Therefore,
for a valid coherence measure defined in \cite{Baumgratz2014}, such as
$C_{l_1}$ and $C_\textrm{re}$, we obtain
\begin{align}
C(\rho)\geq C(\mathcal{P}_n\rho\mathcal{P}_n^\dagger)\geq C(\mathcal{P}_n^\dagger(\mathcal{P}_n\rho\mathcal{P}_n^\dagger)\mathcal{P}_n)=C(\rho),
\end{align}
which indicates that $\mathcal{P}_n$ is a coherence-preserving operator. On the other hand, one can identify the decoherence effect of the \textit{relabeling operator}
$\mathcal{R}_n$ by acting it on the off-diagonal elements $|i\rangle\langle j|$
\begin{align}
\mathcal{R}_n|i\rangle\langle j|\mathcal{R}_n^\dagger=|f_n(i)\rangle\langle f_n(j)|.
\end{align}
When $i=j$ we have $f_n(i)=f_n(j)$ for all $n$ and probably $f_n(i)$ may not be equal to $i$.
This means that $\mathcal{R}_n$ may transfer a diagonal element to the other position on the diagonal.
If $i\neq j$ two possible cases emerge: (i) $f_n(i)\neq f_n(j)$, a situation in which the coherence
is retained but the position of this element is accordingly changed; (ii) $f_n(i)=f_n(j)$,
which implies that $f_n$ is not injective (i.e., many-to-one) and the $|i\rangle\langle j|$-coherence is destroyed.
In contrast to $\mathcal{P}_n$, $\mathcal{R}_n$ could be a coherence-destroying operator.

On the other hand, if we only focus on Kraus operators for GIO we obtain
\begin{align}
K^{\textrm{GIO}}_n|i\rangle\langle j|K^{\textrm{GIO}\dagger}_n=c^{(n)}_ic^{(n)\ast}_j|i\rangle\langle j|,
\end{align}
with $|c^{(n)}_ic^{(n)\ast}_j|\leq1$. Therefore, a GIO, or equivalently, a correlation matrix $\mathcal{C}$
can be regarded as a particular \textit{square sieve} for density matrices, since it preserves the diagonal entries
but partially obstructs the off-diagonal elements. This analogy reflects the unspeakable nature of GIO.
However, for SIO and IO, while $K^{\textrm{GIO}}_n$ is mainly responsible for coherence-destruction,
the permutation operator $\mathcal{P}_n$ and relabeling operator $\mathcal{R}_n$ enable
the transfers between different incoherent basis states. In fact, the above analysis implies the reason why GIO or SIO is
equally powerful as other seemingly more powerful operations (such as IO or MIO) on many occasions, a phenomenon emerged in many
recent relevant works \cite{Chitambar2016b,Zhu2017,Bu2017}.

In view of the above general consideration, we can also make explicit the structure of the interaction unitary
operators $\mathbb{U}$ for SIO and IO. Here we can adopt the method present in \cite{Buscemi2003},
where $\mathbb{U}$ can be constructed by a series of orthogonal isometries
\begin{align}
\mathbb{U}=V\otimes\langle a_0|+\sum_{i=1}^{d_A-1}W_i\otimes\langle g_i|,
\label{Structure}
\end{align}
where $\{|a_0\rangle,|g_1\rangle,\ldots|g_{d_A-1}\rangle\}$ constitutes another orthogonal basis for the Hilbert space $\mathcal{A}$
of the apparatus system and the set of isometries $\{V,W_1,\ldots W_{d_A-1}\}$ is orthogonal to each other.
Note that $V$ is of the form $\sum_iK_i\otimes|a_i\rangle$ and $\{W_i\}$ can be obtained by a repeated use of
the Gram-Schmidt method \cite{Buscemi2003}. Furthermore, the orthogonality of the set of isometries
(i.e., $V^\dagger W_i=0$ and $W_i^\dagger W_j=\delta_{ij}\openone_{\mathcal{H}}$) leads to
the fact that the ranges of distinct isometries are disjoint
and hence the unitarity of $\mathbb{U}$ is easily checked.
Moreover, when restricted to the subspace $\mathcal{H}\otimes|a_0\rangle\langle a_0|$,
the corresponding \textit{effective} $\mathbb{U}$ only contains the first term in Eq. (\ref{Structure}) \cite{Kraus1983},
which is of the form
\begin{align}
\mathbb{U}_{\textrm{SI}}=\sum_{ni}c^{(n)}_i|\pi_n(\phi_i)\rangle\langle\phi_i|\otimes|a_n\rangle\langle a_0|,\\
\mathbb{U}_{\textrm{IO}}=\sum_{ni}c^{(n)}_i|f_n(\phi_i)\rangle\langle\phi_i|\otimes|a_n\rangle\langle a_0|.
\end{align}
It should be emphasized that technically $\mathbb{U}_{\textrm{SI}}$ and $\mathbb{U}_{\textrm{IO}}$ are not unitary operators
(e.g., can be extended to a proper unitary operator by the above procedure) and the constraints on $c^{(n)}_i$ are also different.
For SIO, $c^{(n)}_i$ are restricted such that
the vectors $|c_i\rangle=(c_i^{(1)},c_i^{(2)},\ldots,c_i^{(r)})$ are normalized, which is equivalent to the case of GIO.
However, for IO, the constraint is fully characterized by
\begin{align}
\sum_{n:f_n(i)=f_n(j)}c_i^{(n)\ast} c_j^{(n)}=\delta_{ij}.
\end{align}

\section{DISCUSSION AND CONCLUSION}\label{sec5}
In this paper, we try to establish a comprehensive connection between coherence measures and conventional
decoherence processes. As an example, the most obvious consequences of the von Neumann measurement
are the complete elimination of off-diagonal elements (with respect to the basis specified by the spectrum of an observable) \cite{Zurek1981,Zurek2003}
and the entropy increase of the observed system \cite{Neumann}. It signifies that these phenomena can be employed to define
the valid coherence measures, even prior to the rigorous mathematical framework of Ref. \cite{Baumgratz2014},
where $C_{l_1}(\rho)$ and $C_{\textrm{re}}(\rho)$ are proposed as popular measures of coherence.

Inspired by work in Ref. \cite{Marvian2016}, we have extended our discussion to the L\"{u}ders-type measurement and proposed generalized
coherence measures $C_{l_1}(R,\rho)$ and $C_{\textrm{re}}(R,\rho)$ for possibly degenerate observable $R$,
which, by its eigen-decomposition $R=\sum_nr_nP_n$, splits the Hilbert space into degenerate subspaces.
Note that the L\"{u}ders-type state transformation formula can be derived from
the assumptions of discreteness of spectrum, eigenvalue-repeatability, and minimum disturbance principle.
Among these, repeatability hypothesis is indeed equivalent to the requirement that the transformed state
should belong to the set of generalized incoherent states (i.e., of block-diagonal structure $\rho=\sum_nP_n\rho P_n$),
while the minimum disturbance principle will further select the \textit{closest} one form the geometric point of view.

It is worth emphasizing that the $l_1$ norm of coherence is sensitive to the choice of eigen-decompositions of eigenspaces
characterized by $P_n$, which is tantamount to specifying a fine-grained nondegenerate observable $\mathfrak{R}$ satisfying $f(\mathfrak{R})=R$.
This is exactly the von Neumann's treatment when facing the degenerate observable. In contrast,
the relative entropy of coherence $C_{\textrm{re}}(R,\rho)=S(\rho\|\mathcal{L}(\rho))$
is free from this trouble, and meanwhile highlights the interpretation
that coherence can be regarded as a sort of \textit{incompatibility information} since \cite{Herbut2005}
\begin{align}
[\rho,R]=0\Leftrightarrow \rho=\mathcal{L}(\rho)=\sum_nP_n\rho P_n.
\end{align}
Moreover, compared to the von Neumann measurement, the L\"{u}ders-measurement-dependent quantum discord $\delta^{A|B}(R,\rho^{AB})$
(the observable R acting on subsystem B)
can be also formulated as the difference between the coherence in the global and local states \cite{Yadin2016}. An obvious
sufficient condition for $\delta^{A|B}(R,\rho^{AB})=0$ is the compatibility of $\rho_{AB}$ and $R$, i.e., $[\rho_{AB},R]=0$.
However, the necessary and sufficient condition for zero L\"{u}ders-type discord is left as an open question.

Since the von Neumann and L\"{u}ders measurements are special cases for GIO, we present a detailed analysis of
the structure and physical relation of GIO. We illustrate that GIO is the core of SIO and IO by introducing
the permutation operator $\mathcal{P}_n$ and relabeling operator $\mathcal{R}_n$. In fact, a GIO can be viewed
as a particular sieve which preserves the elements on the main diagonal but partially blocks the off-diagonal positions.
This implies that the Kraus operators of SIO and IO can be constructed by combining a Kraus operator of diagonal form (which
we can call the GIO part) with $\mathcal{P}_n$ or $\mathcal{R}_n$ respectively,
and the decoherence effect are mainly induced by the corresponding GIO part.
This is exactly what the word ``core'' means in the Abstract.

Another problem attracting our attention is the implication of repeatability for a measurement of
a discrete sharp observable. Indeed, in a system-apparatus measurement model of a discrete degenerate observable $R=\sum_nr_nP_n$,
the bipartite interaction unitary operator $\mathbb{U}$ is of the form
\begin{align}
\mathbb{U}\left(|\phi_{ni}\rangle\otimes|a_0\rangle\right)=|\theta_{ni}\rangle\otimes|a_n\rangle,
\end{align}
where the vectors $\{|\phi_{ni}\rangle\}$ form a orthogonal basis of $\mathcal{H}$ such that $R|\phi_{ni}\rangle=r_n|\phi_{ni}\rangle$
and $\{|\theta_{ni}\rangle\}$ is \textit{any} set of normalized vectors in $\mathcal{H}$ satisfying the orthogonality
conditions $\langle\theta_{ni}|\theta_{nj}\rangle=\delta_{ij}$ for all $i,j$ and any $n$ \cite{Beltrametti1990}.
Obviously, for the L\"{u}ders measurement, the choice of the set $\{|\theta_{ni}\rangle\}$ is just $\{|\phi_{ni}\rangle\}$.
However, if we only require the measurement to satisfy the repeatability condition, it is equivalent to require that
$P_n|\theta_{ni}\rangle=|\theta_{ni}\rangle$ for all $i$, which means that $|\theta_{ni}\rangle$ lies within the eigenspace corresponding to $r_n$
and $\{|\theta_{ni}\rangle\}$ constitute another orthogonal basis of $\mathcal{H}$ (see Lemma 1 in \cite{Buscemi2004} or discussions in \cite{Luders}).
Therefore, for an initial state $|\phi\rangle=\sum_{ni}\alpha_{ni}|\phi_{ni}\rangle$, the final states induced by the L\"{u}ders measurement
and this more general repeatable measurement are given by
\begin{align}
\rho_1&=\sum_nP_n\rho P_n=\sum_{n,i,j}\alpha_{ni}\alpha_{nj}^\ast|\phi_{ni}\rangle\langle\phi_{nj}|,\\
\rho_2&=\sum_nK_n\rho K_n^\dagger=\sum_{n,i,j}\alpha_{ni}\alpha_{nj}^\ast|\theta_{ni}\rangle\langle\theta_{nj}|,
\end{align}
with the Kraus operator $K_n=\sum_i|\theta_{ni}\rangle\langle\phi_{ni}|$. Intriguingly, if the residual coherences of final states
are defined in their respective basis, we have
\begin{align}
C_{l_1}^{\{|\phi_{ni}\rangle\}}(\rho_1)&=C_{l_1}^{\{|\theta_{ni}\rangle\}}(\rho_2)=\sum_n\sum_{i\neq j}|\alpha_{ni}\alpha_{nj}^\ast|,\\
C_{\textrm{re}}^{\{|\phi_{ni}\rangle\}}(\rho_1)&=C_{\textrm{re}}^{\{|\theta_{ni}\rangle\}}(\rho_2)=S(\{|\alpha_{ni}|^2\})-S,
\end{align}
where $S(\{|\alpha_{ni}|^2\})$ is the Shannon entropy of the probability distribution $\{|\alpha_{ni}|^2\}$ and $S=S(\rho_1)=S(\rho_2)$.
Therefore, the repeatability condition simply guarantees that the (properly defined) residual coherence contained in the final state is identical to that of
the L\"{u}ders measurement.

Finally, we notice that a more general notion of coherence is proposed recently for a positive operator-valued measure (POVM)
$\mathcal{M}=\{M_n\}$ with $M_n\geq0$ and $\sum_nM_n=\openone$ \cite{Coles2016}
\begin{align}
C_{\textrm{G}}(\rho)=S\left(\rho\|\sum_nM_n\rho M_n\right).
\end{align}
Note that $C_{\textrm{G}}(\rho)$ is well defined (i.e., $C_{\textrm{G}}(\rho)\geq0$) due to the fact $\textrm{Tr}(\sum_nM_n\rho M_n)\leq1$.
This quantity is involved in the derivation of key rates for unstructured quantum key distribution protocols \cite{Coles2016}.
However, since in general $\sum_nM_n\rho M_n$ is not normalized, one may define a modified version
by introducing the generalized L\"{u}ders operations \cite{Busch1991}
\begin{align}
\widetilde{C}_{\textrm{G}}(\rho)=S\left(\rho\|\sum_nM_n^{1/2}\rho M_n^{1/2}\right).
\end{align}
Similarly, we have $\widetilde{C}_{\textrm{G}}(\rho)\geq0$ but its physical meaning and application are left for
future investigation.

\begin{acknowledgments}
This research is supported by the Science Challenge Project (Grant No. TZ2017003-3) and
the National Natural Science Foundation of China (Grant No. 11605166).
C.P.Sun also acknowledges financial support from the National 973 program (Grant No. 2014CB921403),
the National Key Research and Development Program (Grant No. 2016YFA0301201),
and the National Natural Science Foundation of China (Grants No. 11421063 and 11534002).
\end{acknowledgments}
\appendix

\section{Repeatability hypothesis}\label{A1}
Since von Neumann's measurement scheme can be viewed as a particular case of L\"{u}ders postulate,
here we only need to consider the implication of repeatability hypothesis on the state transformation of L\"{u}ders-type measurement.
Let $R$ be a (degenerate) Hermitian operator with the discrete spectral form
\begin{align}
R=\sum_nr_nP_n,
\end{align}
where $r_n$ are distinct eigenvalues and $\sum_nP_n=\openone$ with $\textrm{Tr}(P_n)\geq1$.
Before proceeding, we may employ a useful lemma first proved by von Neumann \cite{Neumann}.

\begin{lemma}
For positive semi-definite operators $A\geq0$ and $B\geq0$, we have $AB=0$ if $\textrm{Tr}(AB)=0$.
\end{lemma}
\begin{proof}
Since $A\geq0$ and $B\geq0$, we have the following
\begin{align}
\textrm{Tr}(AB)=\textrm{Tr}[(\sqrt{A}\sqrt{B})^\dagger(\sqrt{A}\sqrt{B})]=\|\sqrt{A}\sqrt{B}\|_2^2,
\end{align}
where $\|\bullet\|_2$ denotes the Hilbert-Schmidt norm. If $\textrm{Tr}(AB)=0$, then we get $\sqrt{A}\sqrt{B}=0$
and hence we have $AB=\sqrt{A}\sqrt{A}\sqrt{B}\sqrt{B}=0$.
\end{proof}

Assume that a measurement of observable $R$ on the initial state $\rho$ yields an eigenvalue $r_n$ and the corresponding (normalized) state
after the measurement is given by $\rho_n$. According to the repeatability hypothesis, we have the conditional probability
for an immediate successive measurement of $R$
\begin{align}
P(r_m|r_n)=\textrm{Tr}(\rho_nP_m)=\delta_{mn}.
\end{align}
In particular, we obtain $\textrm{Tr}[\rho_n(1-P_n)]=0$. By utilizing the above lemma, we finally have
$\rho_n=\rho_nP_n=P_n\rho_n=P_n\rho_nP_n$, which means that $\rho_n$ lies in the eigenspace characterized by $P_n$.
Note that if the observable $R$ is nondegenerate, then $P_n$ is a rank-one projection operator
and hence $\rho_n=P_n=|\phi_n\rangle\langle\phi_n|$.

Further, based on Born's statistical rule, the initial state $\rho$ is transformed to a statistical mixture of
the sub-ensembles
\begin{align}
\sigma=\sum_nP(r_n)\rho_n=\sum_nP(r_n)P_n\rho_nP_n.
\end{align}
Since $P_mP_n=\delta_{mn}$, we have
\begin{align}
\sigma=\sum_nP_n\left[\sum_mP(r_m)\rho_m\right]P_n
=\sum_nP_n\sigma P_n.
\end{align}
This indicates that the L\"{u}ders measurement transforms the initial state $\rho$ into $\sigma$ with a block-diagonal structure.

\section{Entropy increase for unital channels}\label{A2}
The phenomenon of entropy-increase in the (one-dimensional) projection measurement was first recognized by von Neumann \cite{Neumann}.
Generally, it is easy to prove that the L\"{u}ders-type measurements $\mathcal{L}(\rho)=\sum_nP_n\rho P_n$ increase the von Neumann entropy by Klein's inequality since
\begin{align}
S(\mathcal{L}(\rho))-S(\rho)=S(\rho\parallel\mathcal{L}(\rho))\geq0,
\end{align}
where $S(\rho\|\sigma)=\textrm{Tr}(\rho\log\rho)-\textrm{Tr}(\rho\log\sigma)$ is the quantum relative entropy.

Moreover, there is another elegant way to gain more insight into this fact. Especially, for the von Neumann
measurement of density matrix $\rho$ (where $P_n$ are rank-one orthogonal projectors), the Schur-Horn's theorem
leads to the following majorization relation \cite{Horn1954}
\begin{align}
\bm{\lambda}(\mathcal{D}(\rho))\prec\bm{\lambda}(\rho),
\end{align}
where $\bm{\lambda}(\rho)$ denotes the vector of eigenvalues of $\rho$. For more general cases,
we note that there exists a unitary mixing representation of the pinching operation $\mathcal{L}(\rho)$
\begin{align}
\mathcal{L}(\rho)=\sum_{n=1}^NP_n\rho P_n=\frac{1}{N}\sum_{k=1}^NU_k\rho U_k^\dagger,
\end{align}
where $N$ is the number of elements of the set $\{P_n\}$, which corresponds to the distinct eigenvalues
of the observable $R=\sum_nr_nP_n$, and the unitary matrix $U_k$ is defined as
\begin{align}
U_k=\sum_{j=1}^N\omega^{jk}P_j, \, \omega=e^{2\pi i/N}.
\end{align}
Therefore, according to Alberti-Uhlmann's theorem \cite{Alberti}, we have
\begin{align}
\bm{\lambda}(\mathcal{L}(\rho))\prec\bm{\lambda}(\rho).
\end{align}

Since the von Neumann entropy is a symmetric concave function (then automatically Schur-concave),
we obtain $S(\mathcal{L}(\rho))\geq S(\rho)$. This fact can also be confirmed directly by
the concavity of entropy using the unitary mixing representation of $\mathcal{L}(\rho)$
\begin{align}
S(\mathcal{L}(\rho))=S(\frac{1}{N}\sum_{k=1}^NU_k\rho U_k^\dagger)\geq\frac{1}{N}\sum_{k=1}^NS(U_k\rho U_k^\dagger)=S(\rho).
\end{align}
It is easy to see that $\mathcal{D}(\rho)$ and $\mathcal{L}(\rho)$ are both unital channels. In fact,
the similar majorization relation holds for all unital channels $\Phi(\openone)=\openone$,
i.e., $\bm{\lambda}(\Phi(\rho))\prec\bm{\lambda}(\rho)$ \cite{Alberti,Ando1989}.
Besides, the increase of entropy for unital channels can also be proved by the monotonicity of quantum relative entropy
under CPTP maps in $d$-dimensional Hilbert space, that is
\begin{align}
S\left(\rho\|\frac{\openone}{d}\right)\geq S\left(\Phi(\rho)\|\Phi(\frac{\openone}{d})\right)
=S\left(\Phi(\rho)\|\frac{\openone}{d}\right),
\end{align}
which is equivalent to $S(\Phi(\rho))\geq S(\rho)$.


\end{document}